\documentclass[lettersize,journal]{IEEEtran}

\usepackage{amsmath,amsfonts}
\usepackage{algorithmic}
\usepackage{algorithm}
\usepackage{array}
\usepackage[caption=false,font=normalsize,labelfont=sf,textfont=sf]{subfig}
\usepackage{textcomp}
\usepackage{stfloats}
\usepackage{url}
\usepackage{verbatim}
\usepackage{graphicx}
\usepackage{cite}
\usepackage{booktabs}
\usepackage{amssymb}
\usepackage{subcaption}

\usepackage{amsthm}

\newtheorem{theorem}{Theorem}
\newtheorem{lemma}{Lemma}

\newtheorem{corollary}{Corollary}

\begin{document}

\title{Optimal Sampling under Cost for Remote Estimation of the Wiener Process over a Channel with Delay}

\author{Orhan T Yavaşcan,~\IEEEmembership{}
        Süleyman Çıtır,~\IEEEmembership{}
        and~Elif Uysal,~\IEEEmembership{Fellow,~IEEE}
        \thanks{O. Yavaşcan was with METU, Ankara, Turkey S. Çıtır and E. Uysal are with Freshdata Technologies, Ankara, Turkey, and with METU. This work was supported in part by TUBITAK Grant 22AG019.}
        {\\
        Department of Electrical and Electronics Engineering, Middle East Technical University, Ankara, Türkiye\\
        orhan.yavascan@metu.edu.tr suleyman.citir@metu.edu.tr, uelif@metu.edu.tr
        }
}

\maketitle

\markboth{Journal of \LaTeX\ Class Files,~Vol.~14, No.~8, August~2021}
{Shell \MakeLowercase{\textit{et al.}}: A Sample Article Using IEEEtran.cls for IEEE Journals}

\begin{abstract}
We address the optimal sampling of a Wiener process under sampling and transmission costs, with the samples being forwarded to a remote estimator over a channel with IID delay. The goal of the estimator is to reconstruct the real-time signal by minimizing a long-term average cost that includes both the mean squared estimation error (MSE) and the costs associated with sampling and transmission from causally received samples. Rather than pursuing the conventional MMSE estimate, our objective is to derive a policy that optimally balances estimation accuracy and resource expenditure, yielding an MSE-optimal solution under explicit cost constraints. We look for optimal online strategies for both sampling and transmission. By employing Lagrange relaxation and iterative backward induction, we derive an optimal policy that balances the trade-offs between estimation accuracy and costs. We validate our approach through comprehensive simulations, evaluating various scenarios including balanced costs, high sampling costs, high transmission costs, and different transmission delay statistics. Our results demonstrate the effectiveness and robustness of the proposed joint sampling and transmission policy in maintaining lower MSE compared to conventional periodic sampling methods. The differences are particularly striking under high delay variability. We also analyze the convergence behavior of the cost function. We believe our formulation and results provide insights into the design and implementation of efficient remote estimation systems in stochastic networks.
\end{abstract}


%
\IEEEpeerreviewmaketitle

\section{Introduction}
%
%
%

Many real-time control and remote monitoring tasks in industrial automation, cyberphysical systems, autonomous vehicles, healthcare monitoring systems, and smart grids rely on state estimation of unstable processes, based on samples obtained over a communication network. Classically, the communication network is responsible for the reliable and to some extent orderly and timely transport of samples that it has been handed by the application. In the emerging goal-oriented networking paradigm,~\cite{Ari2023} the communication network ceases to be agnostic of the eventual purpose of the data, in other words, the network takes responsibility of delivering timely and useful packets or updates to the task of computation at the destination. In real-time machine type communication, usefulness of the data is influenced by the staleness of the data. Data can get stale waiting in queue at various link, network and transport layer interfaces. Therefore, the timeliness and effectiveness of the information transferred by the network is inherently tied by the delays that occur due to congestion, relaying and forwarding decisions, and link scheduling.  

In the face of such delays, ensuring goal-oriented transmission, that is, transmitting data that will be effective at the point of computation by the time it gets there, starts with controlling the data generation process. That is, the sample generation and transmission processes are informed by the network state. For instance, in smart grids, the status information about power demand, supply conditions, and grid health constantly changes. To ensure efficient, reliable, and safe system performance, controllers need accurate estimates of the grid's current status from nearby sensors~\cite{Eapen2023}. This necessitates a decision policy on when to generate and transmit samples, to minimize estimation error. Designing such an optimal sampling strategy has been addressed in the previous literature~\cite{Elif-remote-est, Tang2023} which assumed that the process can be monitored continuously at the sender side.  

In practical scenarios there is cost. For example, in battery-powered IoT sensor networks, frequent sampling drains the battery quickly (sampling cost) while transmitting each sample over constrained wireless links incurs bandwidth or energy costs (transmission cost). In industrial automation, sensors monitoring machine health or production lines face overhead for each measurement (e.g., sensor wear, resource usage) and incur transmission costs in wired or wireless networks. In healthcare monitoring, wearable sensors must strike a balance between accurate, frequent measurements and preserving battery life or limited bandwidth.

Incorporating practical constraints on the cost of sampling and transmission makes the problem more challenging, but it is important to address this generalization for results that are more relevant to the design of networked sensing architectures.

\begin{figure}[ht]
    \includegraphics[width=0.5\textwidth]{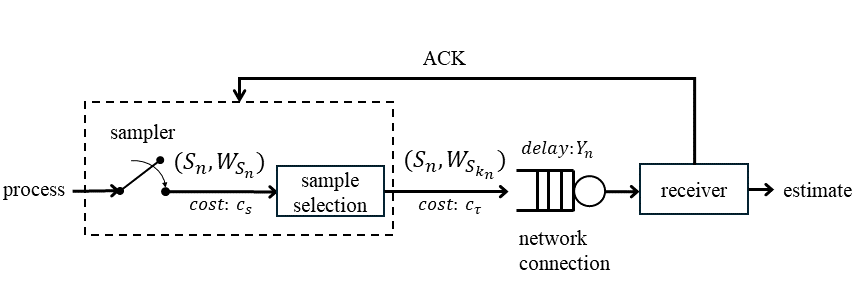}
    \caption{System Model}
    \label{fig:system model}
\end{figure}

In this paper, we study the optimal remote estimation of a Wiener Process from samples that are transmitted through a network connection with random delay, as illustrated in Figure \ref{fig:system model}. The random delay models the packet transmission latency in a networked system, where the estimator has access to the data samples over a wired or wireless network. Delays in such a scenario may be caused by the existence of links with long propagation delay on the path of the data flow, queuing delay due to congestion at various network interfaces, and re-transmissions due to packet drops. 

We will model the source as a Wiener process \(\{W_t, t\in \mathbb{R} \}\), which describes the values of a scalar measured quantity such as displacement, temperature, etc, that has time variation. The Wiener process is a well-known example of a non-stationary stochastic process, characterized by stationary and independent increments. It has broad applications in modeling the state of an unstable system, in engineering and other fields including mathematics, economics, finance and physics. We model a remote estimator which will obtain a minimum mean square error (MMSE) estimate \(\hat{W}_t\) based on causally received samples. We assume that the transmitter sends encoded and time-stamped samples as individual packets\footnote{The design of source and network connection encoders are outside the scope of our treatment.}. In contrast to the previous literature~\cite{Elif-remote-est} on this problem, in this paper we incorporate sampling and transmission costs. In other words, the objective is to find optimal sampling and transmission policies to optimize mean squared error at the estimator, under sampling and transmission costs. In the rest, after a discussion of related work, we will make the problem statement precise and then present our  methodology for finding optimal policies that utilizes Lagrange relaxation and backward induction.

\section{Related Work}
While the optimal sampling of the Wiener process over  a noisy communication channel has been studied in previous literature (eg,~\cite{Rabi2012}) the treatment of delay was first considered in~\cite{Elif-remote-est}. This paper can be considered as an extension of the treatment in~\cite{Elif-remote-est} as it explores optimal sampling of the Wiener process under delays for minimizing mean square estimation error under a sampling frequency constraint, establishing a connection between age of information (AoI) and estimation error. Online sampling of the Wiener process to minimize MSE was also studied in~\cite{tang2022sampling} under unknown transmission delays, proposing an adaptive algorithm with cumulative MSE regret \(O(\ln k)\) and proving its minimax optimality. The sampling of an Ornstein-Uhlenbeck process over a network was considered in\cite{OU-process}. The study showed that without knowledge of the signal (the signal-agnostic solution), minimizing the age of information is optimal, whereas using causal knowledge of signal values can achieve lower estimation errors. The work~\cite{Ari2023} introduces goal-oriented sample selection and scheduling for remote inference over a channel with variable delay that exhibits memory, which can model a Delay and Disruption Tolerant Network (DTN)~\cite{uysal2024}. It is shown that an index-based threshold policy effectively utilizes delayed feedback to minimize the expected time-average inference error. Our work contributes to the aforementioned set of works by incorporating both sampling and transmission costs into the optimization problem. 

Along with those studies the paper can be considered as a contribution to the literature on remote estimation, e.g.,~\cite{gao2015optimal,gao2016remote,imer2010optimal,nar2014sampling,chakravorty2015distortion,lipsa2011remote,nayyar2013optimal}.
\section{System Model and Problem Statement}

We model the following remote estimation problem: There is an observer of the Wiener process \(\{W_t, t\in \mathbb{R} \}\) that can sample the process at time instants \(S_n\), according to a \textit{sampling policy}. We denote the samples by \((S_n, W_{S_n}), n\in \mathbb{Z}, n\geq 1\). 

Upon sampling at time \(S_n\), the pair \((S_n, W_{S_n})\) is made available to the  \textit{transmission policy}, which may decide to discard or transmit this sample. Upon a transmission decision, the sample will be encoded into a \textit{packet} to be given to the transmitter. The transmitter will send this packet to the estimator over a network connection (see Fig. \ref{fig:system model}). 
The actions of sampling and transmission have  associated costs of $c_s>0$ and $ c_\tau>0$, respectively. Costs are deterministic values that represent the resource expenditure for each respective action and are used to guide the optimization process. Note that, unlike in the original formulation in ~\cite{Elif-remote-est}, sender i.e. the transmission policy does not have access to the process $W(t)$ for all time, and can only observe the value of the samples. It implies that multiple samples may need to be produced before the next one that is sent.
In the rest, we will combine the sampling and transmission policies to a single jointly optimized policy. Note that there is a trade-off between the cost of sampling and the benefit of observing the process frequently. On the other hand, the existence of a non-zero transmission cost will lead to some of the generated samples ending up being discarded rather than transmitted. Through these costs, our problem formulation models some real-life implementation trade-offs that were abstracted out in ~\cite{Elif-remote-est}. 

We shall use the notation that the \(n^{th}\) packet to be transmitted is the \(k_n^{th}\) sample taken. Note that there are possibly multiple samples taken between any two samples chosen for transmission, such that $k_{n+1}-k_n\geq 0$. Also note that the transmission time for a sample might be larger than its sampling time, as the scheduling policy may decide to hold the head-of-line packet in the buffer by some time before deciding upon making a transmission. Figure \ref{fig:timeline} illustrates the process timeline. 

The estimator constructs the signal $\hat{W}_t$, which is the MMSE estimate of the original process based on the samples received by time $t$. The sampling and transmission policies are, however, designed under cost constraints.

\begin{figure}[ht]
\includegraphics[width=0.5\textwidth]{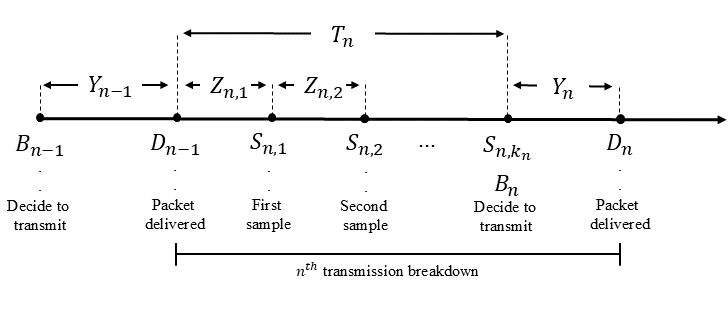}
    \caption{Process Timeline}
    \label{fig:timeline}
\end{figure}

We assume that the packets experience i.i.d. random transmission delays over the network connection. Specifically, the delay experienced by the $n^{th}$ packet is $\{Y_n, n\geq 1\}$, where $Y_n\geq 0$ is a random variable satisfying $E[Y_n^2] < \infty$. Let $B_n$ represent the time at which the transmission of the $n^{th}$ packet begins. Delivery time of packet $n$ is then $D_n = B_n + Y_n$ and $S_{n,j}$ is the time of the $j^{th}$ sample prior to the transmission of $n^{\rm{th}}$ packet,  $j \in {1,2,...,k_n}$. The initial value $W_0 = 0$ is initially known by the estimator, which is represented by $S_0 = D_0 = 0$. At any time $t$, the estimator forms an estimate $\hat{W}_t$ using the samples received up to time $t$.  \\ The loss will be measured by the time-average mean squared error (MSE):
\begin{equation}
\text{mse} = \limsup_{T \to \infty} \frac{1}{T} E \left[ \int_{0}^{T} (W_t - \hat{W}_t)^2 dt \right]. \quad 
\end{equation}

Here $\hat{W}_t$ is the MMSE estimate at time $t$ which uses the sampling time $S_j$, sample value $W_{S_j}$, delivery times $\{D_j\}$, number of the samples $M(t)$ delivered by time t (i.e. number of packets), and the fact that no sample has been received after the latest sample delivery, captured by the following information set:

\begin{equation}
\mathcal{I}(t) = \{(S_j, W_{S_j}, D_j)\} \cup \{M(t) = \max\{j \, | \, D_j \leq t\}\}.
\label{information-set}
\end{equation}

\noindent
This definition accounts for all samples created between delivery events, regardless of whether or not they were eventually transmitted.

The estimator assumes no additional information is inferred during silent intervals, meaning the absence of new deliveries does not trigger updates or adjustments. This simplifies the remote estimation problem by relying solely on the arrived samples. While “no new delivery” could imply partial information (e.g., no triggering event or transmission failure), ignoring this preserves a clean framework. Consequently, the time-average mean squared error (MSE) is minimized using the classical minimum mean-square error (MMSE) strategy \cite{poor-detection}. Putting all together; For $t \in [D_n, D_{n+1})$, the best available information is the most recently delivered sample, and by design, the estimator remains unaffected by the absence of deliveries between $D_n$ and $D_{n+1}$. Then the MMSE estimate $\hat{W}_t$ of $W_t$ is:
\begin{subequations}
\begin{align}
\hat{W}_t &= \mathbf{E} \left[ W_t | \mathcal{I}(t)  \right] \\
&=\mathbf{E} \left[ W_t | (S_{k_n}, W_{S_{k_n}}), D_n, M(t)  \right] \\
&= \mathbf{E} \left[ W_t | W_{S_{k_n}},M(t) \right]  \\
&= \mathbf{E} \left[ W_t | W_{S_{k_n}} \right] \\
&= W_{S_{k_n}}, \quad if \quad t \in \left[ D_n , D_{n+1} \right], \quad n = 0, 1, 2, . . . ,
\end{align}
\end{subequations}

From 2b to 2c definiton of $M(t)$ is used and from 2c to 2d the fact that conditioning on \(\{ M(t) = n \}\) is effectively the same as conditioning on 
\(\{\,W_{S_{k_n}}\text{ is the last arrived sample up to }t\}\) is used. Lastly, for a Wiener process, \(W_t\) can be expressed as:
\[
W_t = W_{S_{k_n}} + (W_t - W_{S_{k_n}}),
\]
where \(W_{S_{k_n}}\) is the sample value at the last sampling instant \(S_{k_n}\), \(W_t - W_{S_{k_n}} \sim \mathcal{N}(0, t - S_{k_n})\) is the Wiener increment, which is independent of \(W_{S_{k_n}}\). Thus, the conditional expectation is \(\mathbb{E}[W_t \mid W_{S_{k_n}}] = W_{S_{k_n}},\) because the increment \(W_t - W_{S_{k_n}}\) has zero mean and is independent of the past.

In this policy, once a new sample arrives at time $D_n$, the estimator updates to the latest sample $W_{S_{k_n}}$, and then \emph{continues to hold} that value until the next delivery occurs at time $D_{n+1}$.

The sampling times \(S_n\) are defined as stopping times based on the information available up to that point, which includes the realization of the delay process \(Y_j\) experienced up to time \(S_n\). A stopping time is essentially a decision-making point determined by the path of the Wiener process up to that time, ensuring the causality of the sampling policy. This means that at each sampling instant \(S_i\), the decision is made with complete knowledge of the past samples and their transmission delays.

Let $Z_{n,j}$ represent the waiting time between the delivery of sample \(n\) and the generation of sample \(n+1\):\\
\begin{equation}
    Z_{n,j} =
    \begin{cases} 
    S_{n,j}-S_{n,j-1} & \text{if } j > 1 \\
    S_{n,1}-D_{n-1} & \text{if } j=1
    \end{cases}
\end{equation}

These waiting times represent the intervals between sampling instants and are influenced by both the sampling decisions and the transmission delays experienced.  The sampling policy can therefore be expressed as a sequence \((S_1, S_2, \ldots)\) of sampling instants, where each \(S_i\) is generated based on the available information up to that point, including the realized delays \(Y_j\) up to time \(S_i\). If the sequence of delays \((Y_1, Y_2, \ldots)\) is given, and therefore the delivery times \(D_n\) are known, then the waiting times \((Z_1, Z_2, \ldots)\) are uniquely determined by the sequence of sampling times \((S_1, S_2, \ldots)\).

Let the set \(\Pi\) be the set of all \emph{causal joint sampling and transmission policies} such that the waiting times \(\{Z_{n,j}\}_{j=1}^{k_n}\) form a regenerative process. We aim to find a set of decision variables, consisting of the waiting times \(\{Z_{n,1}, \dots, Z_{n,k_n}\}\) and the stopping index \(k_n\), such that:
\[
\bigl\{(Z_{n,1}, \dots, Z_{n,k_n}) ,k_n\bigr\} \in \Pi,
\]
that minimizes the \emph{time-averaged mean squared error (MSE) under cost constraints}.

If inter-sampling waiting time \(Z_{n,j}\) is a regenerative process, then the sum \(T_n = \sum_{j=1}^{k_n} Z_{n,j}\), corresponding to inter-packet waiting time, is also a regenerative process. This is because the regeneration points of \(Z_{n,j}\) serve as regeneration points for \(T_n\), and the segments of \(T_n\) between these points are independent and identically distributed, thereby maintaining the regenerative property.

Because \(T_n\) is a regenerative process, by following the renewal theory in \cite{ross-stochastic} and \cite{haas-petrinets} (See Appendix A. for derivation) the problem for minimizing the time averaged MSE with sampling and transmission costs is formulated as:

\begin{equation}
\text{mse}_{\text{opt}} = \min_{\pi \in \Pi} \!\frac
{
\sum\limits_{n=1}^{\infty} \!\mathbb{E} \!\left[ \int\limits_{D_n}^{D_{n+1}} \!\!(W_t - W_0)^2 \, dt + c_s k_n + c_\tau \right]
}
{
\sum\limits_{n=1}^{\infty} \mathbb{E} \left[ D_{n+1} - D_n \right]
}
\label{eq:mseopt}
\end{equation}

\section{Optimal Policy}
The solution for (\ref{eq:mseopt}) involves Lagrange relaxation of problem, and backward induction. By introducing a Lagrange multiplier \(\lambda\) for normalization, the original problem is reformulated into a relaxed cost function that balances estimation accuracy with sampling and transmission costs.

The optimization is performed over:
\begin{itemize}
    \item \textbf{Sampling Intervals} (\(Z_{n,j}\)): Waiting times between consecutive samples within a transmission epoch.
    \item \textbf{Stopping Index} (\(k_n\)): Number of samples taken before the \(n\)-th transmission.
\end{itemize}

By iteratively adjusting \(\lambda\), the optimal policy is identified, minimizing the long-term average cost under the given constraints.

\subsection{Lagrange Relaxation and Problem Simplification}

To solve the (\ref{eq:mseopt}), we introduce a Lagrange multiplier $\lambda$ and define $\lambda^*$ as the minimum value of the objective function. Noting that $D_{n+1}\!-\!D_n = Y_n+T_n$ redefine the problem to find the optimal policy for an arbitrary $\lambda$ by minimizing the following cost function:

\begin{equation}
J(\lambda)=\lim_{N\to\infty}\frac{1}{N} \sum_{n=1}^{N} C_n(\lambda)
\label{eq:relaxed_problem}
\end{equation}

Where $C_n$, is the expected cost in one delivery epoch:

\[
    C_n(\lambda) = \mathbb{E}\left[\int\limits_{D_n}^{D_{n+1}}(W_t-W_0)^2 dt + c_sk_n + c_\tau -\lambda(Y_n + T_n)\right]
\]

For any policy \( \pi \in \Pi\), we can write:
\begin{equation}
\sum_{n=1}^{\infty}\mathbb{E}\!\left[\int\limits_{D_n}^{D_{n+1}}\mkern-10mu(W_t\!-\!W_0)^2\,dt\!+\!c_s k_n\!+\!c_\tau\!-\!\lambda^*(Y_n\!+\!T_n)\right]\!\geq\!0
\label{eq:lagrangae_multiplier}
\end{equation}

If we let  ${\displaystyle \lambda }$ be nonnegative weight, system get penalized if it violates the constraint, and system are also rewarded if it satisfies the constraint strictly. The above problem is Lagrangian relaxation of our original problem ensuring  the overall cost remains non-negative. 

\begin{lemma}
The following assertions are true for  (\ref{eq:mseopt}) and (\ref{eq:relaxed_problem}):
\begin{itemize}
    \item[(a)] \(\text{mse}_{\text{opt}} \geq \lambda\) if and only if \(J(\lambda) \geq 0\).
    \item[(b)] \(\text{mse}_{\text{opt}} \leq \lambda\) if and only if \(J(\lambda) \leq 0\).
    \item[(c)] If \(J(\lambda) = 0\), the solutions to (\ref{eq:mseopt}) and (\ref{eq:relaxed_problem}) are identical.
\end{itemize}
\label{lemma1}
\end{lemma}

\begin{proof}
See Appendix B.
\end{proof}

Assume we obtain the optimal policy that minimizes \(J(\lambda)\), then we can characterize the behavior of the cost function as follows:
\begin{itemize}
    \item \(J(\lambda^*) = 0\): There exists a value \(\lambda^*\) for which the cost function achieves its minimum value. This \(\lambda^*\) represents the optimal trade-off point where the MSE, sampling cost, and transmission cost are balanced most effectively.
    \item \(J(\lambda) > 0\) if \(\lambda < \lambda^*\): For values of \(\lambda\) less than \(\lambda^*\), the cost function is positive. This indicates that the policy underestimates the required trade-off, leading to a higher overall cost.
    \item \(J(\lambda) < 0\) if \(\lambda > \lambda^*\): For values of \(\lambda\) greater than \(\lambda^*\), the cost function is negative. This suggests that the policy overestimates the required trade-off, resulting in an overall reduction in cost.
\end{itemize}

The optimal policy is derived by iteratively adjusting \(\lambda\) and evaluating \(J(\lambda)\) until the condition \(J(\lambda^*) = 0\) is met. This approach ensures that the selected \(\lambda^*\) provides the most efficient balance between estimation accuracy and the associated costs of sampling and transmission. It is important to note that the state of the system is reset at the beginning of each transmission where state refers to the information available at the start of each transmission epoch, which determines the sampling and transmission decisions. It includes the current estimation error (\(E_m\)), representing the mean square error since the last transmitted sample, and relevant information about the Wiener process (\(Z_{n,j}\), \(Y_n\)). The state is reset at the beginning of each epoch, reflecting the regenerative nature of the system, implies that the decision process for each packet transmission is independent of previous packets, allowing the optimization to be performed on a per-packet basis. Then by value iteration, we derive the optimal waiting times $Z^*$ and associated costs. We use Lemma \ref{lemma2} to simplify the integral part of the cost function.

\begin{lemma}
Let \( X \) and \( Y \) be independent random variables with finite second moments. Then,
\[
\mathbb{E} \left[ \int_{X}^{X+Y} (W_t - W_0)^2 \, dt \right] = \frac{Y^2}{2} + Y(W_X - W_0)^2.
\]
\label{lemma2}
\end{lemma}

\begin{proof}
The proof follows from the properties of the Wiener process. See Appendix C. 
\end{proof}

 By Lemma \ref{lemma2}, by denoting the sampling error as $E_j \overset{\Delta}{=} (w_{Y+\bar{Z_j}}^2 - w_0)^2$,  \eqref{eq:relaxed_problem} can be written as:
 
 \begin{equation}
 \begin{split}
     J(\lambda) & = c_\tau - \lambda\mu_y + \mathbb{E} \left[ \frac{{Z_1}^2}{2} + Z_1 Y  + c_s - \lambda Z_1 \right] \\
     & \qquad + {\sum_{j=2}^{k} \mathbb{E} \left[ \frac{{Z_j}^2}{2} + Z_j {E_{j-1}}  + c_s - \lambda Z_j \right]} \\
     & \qquad + \mu_y \mathbb{E} \left[ E_k \right]
 \end{split}
 \end{equation}

Note that, $\mathbb{E}[E_k] = \mathbb{E}[Y+T_k] = \mu_y + \sum_{j=1}^{k} \mathbb{E}[Z_j] $ where \( \mu_Y \) is the mean transmission delay.  The cost function is then simplified to:
\begin{equation} \label{eq:eliminate-ek}
\begin{aligned}
    J(\lambda) = &c_\tau - (\lambda - \mu_Y)\mu_Y \\
    +& \mathbb{E}[\frac{Y^2}{2}] + \mathbb{E}[\frac{Z_1^2}{2} + Z_1(Y-\lambda + \mu_Y) + c_s] \\
    +& \sum_{j=2}^{k} \mathbb{E}[\frac{Z_j^2}{2} + Z_j(E_{j-1}-\lambda + \mu_Y) + c_s].
\end{aligned}
\end{equation}
Next, we define $h: \mathbb{R}_{\ge0} \times \mathbb{R}_{\ge0} \to \mathbb{R}$ as:
\begin{equation} \label{eq:h-function}
    h(x_1,x_2) \!=\! \frac{(x_1 \!+\! x_2 \!-\! \lambda \!+\! \mu_Y)^2}{2} \!-\! \frac{(x_2 \!-\! \lambda \!+\! \mu_Y)^2}{2} \!+\! c_s.
\end{equation}
and $h_0 \triangleq c_T - (\lambda - \mu_Y)\mu_Y + E[Y^2 / 2]$. 
Then, cost function \eqref{eq:eliminate-ek} can be rewritten in recursive form as:
\begin{equation} \label{eq:cost-simplified}
    h_0 + E[h(Z_1,Y)] + \sum_{j=2}^{k} E[h(Z_j, E_{j-1})].
\end{equation}
\subsection{Finding an optimal policy}
Define $g_r$ as the minimum value of the sum in \eqref{eq:cost-simplified} over a horizon of at most $r$ samples given an initial sampling error, capturing the minimum cost achievable from a specific state onward, helping in the iterative backward optimization process. Optimal waiting times \(\{Z_{n,1}, \ldots, Z_{n,k_n}\}\) for the \(n\)-th packet are determined iteratively using the function \(g_r(E_m)\).

\begin{equation}
    g_r(E_m) 
    = \min_{\{Z_{m+1}, \dots, Z_{m+r}\}}
    \sum_{j=1}^{r} 
    \mathbb{E}\bigl[h\bigl(Z_{m+j}, E_{m+j-1}\bigr)\bigr].
\end{equation}

where $E_{m+j} = \mathcal{N}(\sqrt{E_{m+j-1}}, Z_{m+j})^2$.
Notice that, $g_r$ depends only on $c_s$ and $\lambda - \mu_Y$. 

Suppose we plan to take $r$ additional samples from the current error state $E_m$. We can split this decision as follows: First, choose $Z_{m+1}$, the waiting time until the next sample. Once that sample is taken, the error state transitions to $E_{m+1}$. We now have an $(r-1)$-sample horizon remaining, starting from $E_{m+1}$.

In a recursive form (Bellman recursion), $g_r$ can be expressed as:
\begin{equation}
    g_r(E_m)= \min_{Z_{m+1}}\Bigl\{    \underbrace{h\bigl(Z_{m+1},\, E_m\bigr)}_{\text{current cost}}
      +
      \underbrace{\mathbb{E}\bigl[g_{r-1}(E_{m+1})\bigr]}_{\text{future cost}}
    \Bigr\}
    \label{eq:gr-recursion}
\end{equation}
where $E_{m+1}$ arises from the transition 
$E_{m+1} = \mathcal{N}\bigl(\sqrt{E_m},\,Z_{m+1}\bigr)^2.$
Here, $h\bigl(Z_{m+1}, E_m\bigr)$ represents the immediate expected cost incurred by waiting $Z_{m+1}$ from the error state $E_m$, and $\mathbb{E}[g_{r-1}(E_{m+1})]$ is the expected future cost over the remaining $(r-1)$ steps.

When $r=0$, no further samples are planned, so we set $g_0(E_m)=0$. For any finite $r$, we can compute $g_r(E_m)$ by iterating \eqref{eq:gr-recursion} from $r$ down to 1.

Following, the minimum value of \eqref{eq:cost-simplified} is:
\begin{equation}
    h_0 + E[h(Z_1,Y) + g_{\infty} (E_1)],
\end{equation}
where $g_{\infty} (E_m) = \lim_{r \to \infty} g_r (E_m)$ and $E_1 = \mathcal{N} (0,Y+Z_1)^2$. 

Recall that each new sampling error \(E_{m+j}\) depends on its previous value \(E_{m+j-1}\) and the waiting time \(Z_{m+j}\). This arises from the incremental property of the Wiener process: over any waiting interval \(Z_{m+j}\), the process increment is a zero-mean Gaussian random variable with variance \(Z_{m+j}\). Therefore, if the previous sampling error is \(E_{m+j-1}\), then combining this error with the new increment the error term \(E_{m+j}\) is updated as:
\[
    E_{m+j} = \bigl(W_{t + Z_{m+j}} - W_0\bigr)^2,
\]
where \(t\) is the sampling time of the \((m+j-1)\)-th sample.

\begin{corollary}
    Let $E'$ be the sampling error and let $Z$ be the waiting time until the next sample. Then, probability distribution of the next sampling error, $E''$, is given by:
    \begin{equation}
        E'' = \mathcal{N} (\sqrt{E'},Z).
    \end{equation}
    
\end{corollary}

By modeling the new sampling error \(E''\) in this way, we obtain a \emph{recursive} relationship:
\[
    E_{m+j} 
    = f\bigl(E_{m+j-1}, Z_{m+j}\bigr)
    = \mathcal{N}\bigl(\sqrt{E_{m+j-1}},\,Z_{m+j}\bigr)^2,
\]
where \(f\) represents the function encapsulating the transition from the old error \(E_{m+j-1}\) to the new error \(E_{m+j}\). Corollary 1 highlights this property to emphasize that each new sampling error state can be derived by applying a Gaussian increment to the previous error state. This insight then underpins the dynamic programming framework, wherein the transition from \(E_{m+j-1}\) to \(E_{m+j}\) is accounted for at each step of the backward induction process.

\subsection{Main Theorem}
Given \(E_{m+j-1}\), we update \(E_{m+j}\) iteratively using the mean and variance of the Wiener process. Continue the iteration until \(g_r(E_m)\) converges to \(g_\infty(E_m)\), representing the minimum cost over an infinite horizon. The cost function \(J(\lambda)\) is evaluated iteratively, and the iteration stops when \(J(\lambda)\) becomes sufficiently small (e.g., below a predefined threshold \(\epsilon\)). The stopping criterion for \(k_n\) is based on the marginal contribution of adding another sample \(Z_{n,k_n+1}\): the process halts when the incremental cost of \(h(Z_j, E_{j-1})\) no longer improves the objective \(J(\lambda) = h_0 + \mathbb{E}[h(Z_1, Y)] + \sum_{j=2}^{k} \mathbb{E}[h(Z_j, E_{j-1})]\). This recursive optimization adjusts \(Z_{n,j}\) and determines \(k_n\) inherently, stopping when further samples increase the total cost.

Under appropriate boundedness and Lipschitz conditions on \(h(\cdot,\cdot)\) and a contractive state transition, standard dynamic programming arguments guarantee that the sequence \(\{g_r(E)\}\) converges uniformly to a unique fixed point \(g_\infty(E)\) at a geometric rate \cite{Puterman2005}. Consequently, the overall cost
\[
J(\lambda)= h_0 + \mathbb{E}\Bigl[ h(Z_1,Y) + g_\infty(E_1) \Bigr],
\]
converges to its optimum when \(J(\lambda^*)=0\), yielding the optimal waiting times \(\{Z_{n,j}^*\}\) and stopping indices \(k_n^*\). We formally state this result as follows.

\begin{theorem}[Convergence of the Optimal Policy Algorithm]
Assume that 
\[h(x_1,x_2)= \frac{(x_1+x_2-\lambda+\mu_Y)^2}{2} - \frac{(x_2-\lambda+\mu_Y)^2}{2} + c_s\] 
is bounded and Lipschitz continuous and that the state transition \(E_{m+1} = \mathcal{N}(\sqrt{E_m},\,Z_{m+1})^2\) is contractive. Then the sequence \(\{g_r(E)\}\) defined by
\[
g_r(E) = \min_{Z\ge0} \Bigl\{ h(Z,E) + \mathbb{E}\bigl[g_{r-1}(E')\bigr] \Bigr\}, \quad g_0(E)=0,
\]
converges uniformly to a unique fixed point \(g_\infty(E)\) for all \(E\ge0\) with a geometric rate. Consequently, the overall cost function
\[
J(\lambda)= h_0 + \mathbb{E}\Bigl[ h(Z_1,Y) + g_\infty(E_1) \Bigr],
\]
converges to its optimal value when \(J(\lambda^*)=0\), ensuring that the optimal policy—yielding optimal waiting times \(\{Z_{n,j}^*\}\) and stopping indices \(k_n^*\)—minimizes the long-term average cost that balances estimation accuracy with sampling and transmission costs.
\end{theorem}

  \section{Numerical Analysis}
To provide a comprehensive evaluation of the proposed optimal sampling and transmission policy, we consider a variety of scenarios that cover somewhat mild as well as more extreme operating conditions. These cases help illustrate the behavior of the system under different cost structures and transmission conditions. The results depicted in Figure \ref{fig:mse_vs_delay_var}, comparing the MSE of our joint error and cost optimal policy with a conventional periodic sampling approach with respect to varying delay variance $(\sigma^2)$.

\begin{figure}[ht]
    \centering
    \includegraphics[width=0.45\textwidth]{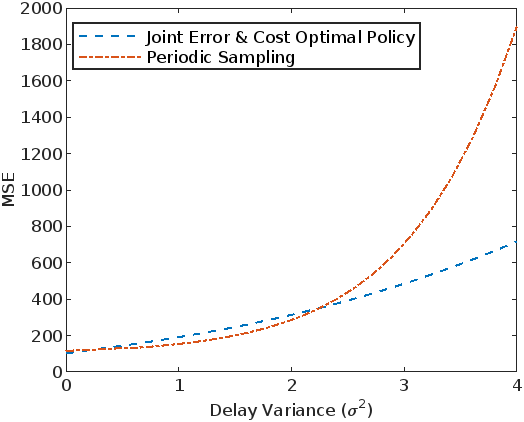}
    \caption{MSE vs. the delay variance $\sigma^2$ for $c_s = 2$, $c_t = 5$, and $\lambda = 10$}
    \label{fig:mse_vs_delay_var}
\end{figure}

The optimal policy was derived using Lagrange relaxation and iterative backward induction to minimize the time-average MSE, considering the costs associated with sampling and transmission. The optimal policy consistently achieves significant gain over periodic sampling across different costs, particularly in scenarios with higher delay variability. As the parameter $\sigma$ grows, the tail of the delay distribution becomes heavier. We observe that $\text{MSE}_{\text{periodic}}$ grows quickly with respect to $\sigma^2$, much faster than $\text{MSE}_{\text{opt}}$. Consequently, the estimator's performance degrades significantly, resulting in higher MSE. Conversely, the optimal policy exhibits a more gradual increase due to its adaptive nature, adjusting sampling intervals based on the current state and transmission delays. By balancing the trade-off between sampling frequency and delay variability, the optimal policy maintains lower estimation errors across a broader range of delay conditions.

\begin{figure}[ht]
    \centering
    \includegraphics[]{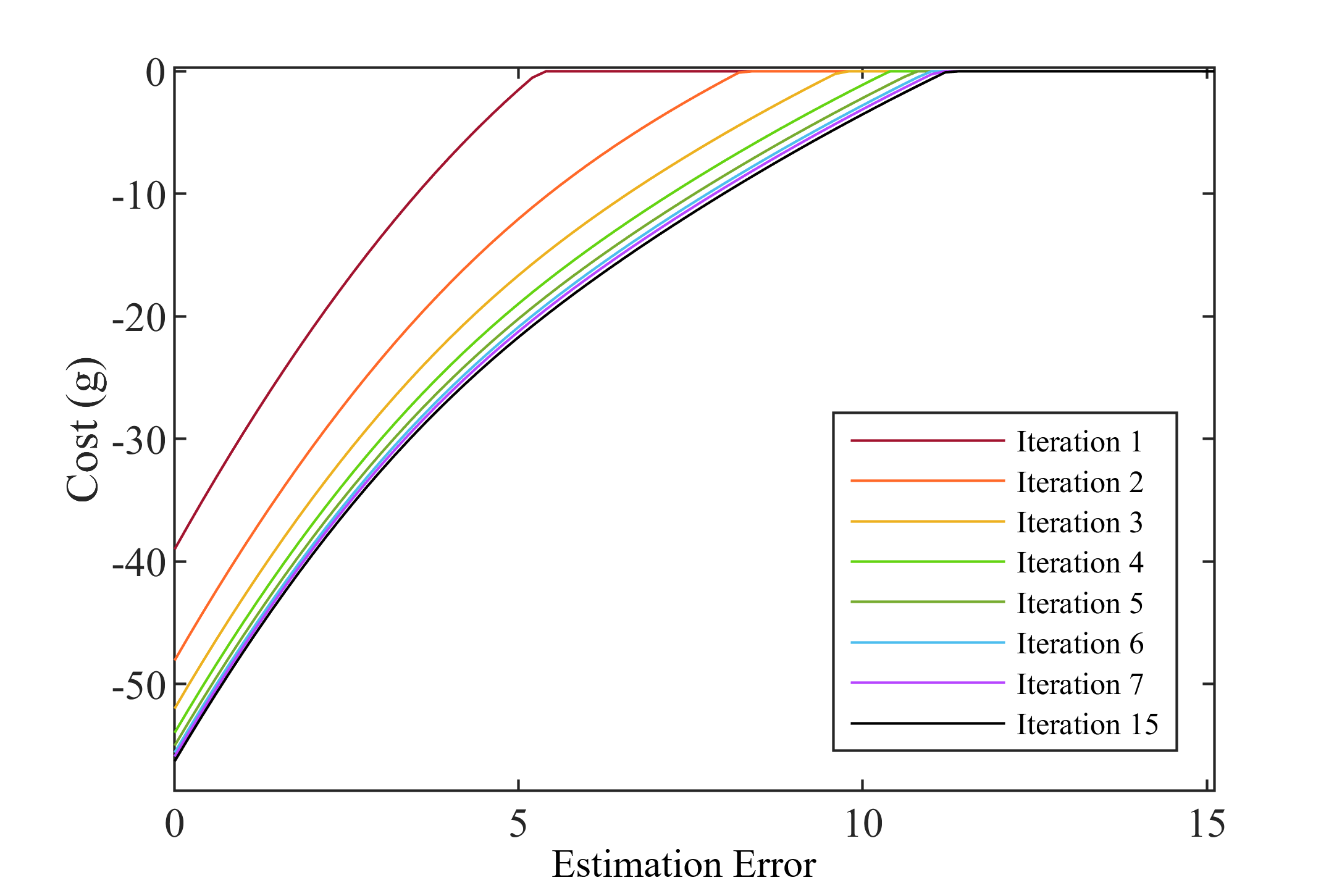}
    \caption{Convergence of cost over iterations for $\sigma^2 = 0.1$ and $c_s = 2$, $c_t = 5$, $\lambda = 10$}
    \label{fig:convergence}
\end{figure}

Figure \ref{fig:convergence} shows the convergence of the \(g\) values over the iterations, indicating that the optimization process is functioning correctly and the algorithm is finding the optimal policy. 
\begin{figure}[ht]
    \centering
    \includegraphics[]{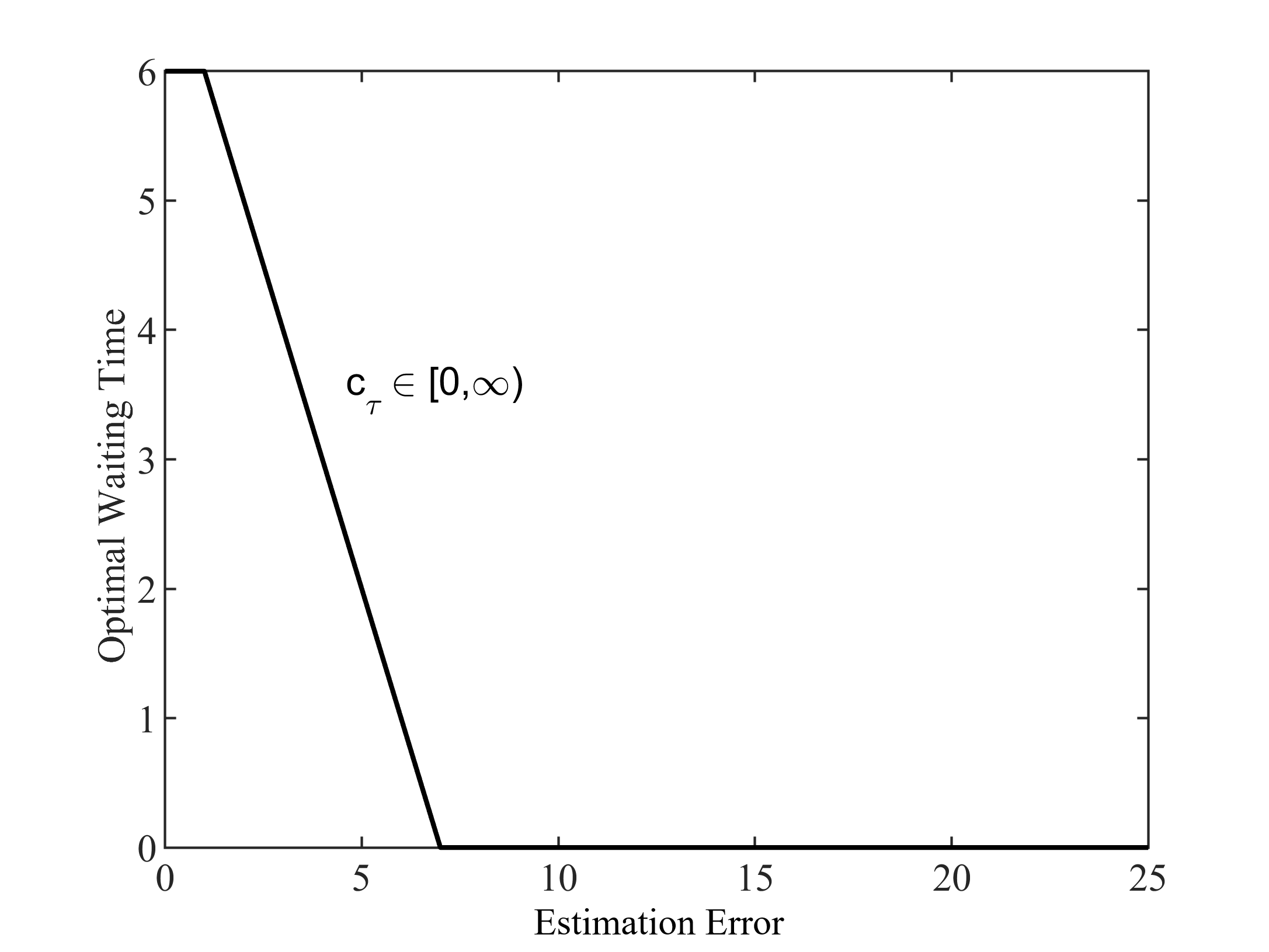}
    \caption{Optimal waiting times vs. error for $\sigma^2 = 0$ and $c_s = 1$ as constant}
    \label{fig:waiting_varY_0}
\end{figure}
In the figure \ref{fig:waiting_varY_0} with \( \sigma^2 = 0 \), the optimal waiting time (\( Z^*\)) decreases linearly with increasing estimation error (E), and the curves for different transmission costs (\( c_\tau = 0.1, 1, 10 \)) overlap completely. This is because the transmission delay \( Y \) is deterministic, eliminating any variability in transmission times. Consequently, the system can directly reduce the waiting time as the error increases to correct the estimation more frequently. The independence of the curves from the transmission cost indicates that, without delay variability, the primary concern is correcting the error. The linear decrease in the optimal waiting time $Z^*$ with increasing estimation error suggests a direct relationship between the error and the time to wait before taking the next sample. 

\begin{figure}[ht]
    \centering
    \includegraphics[]{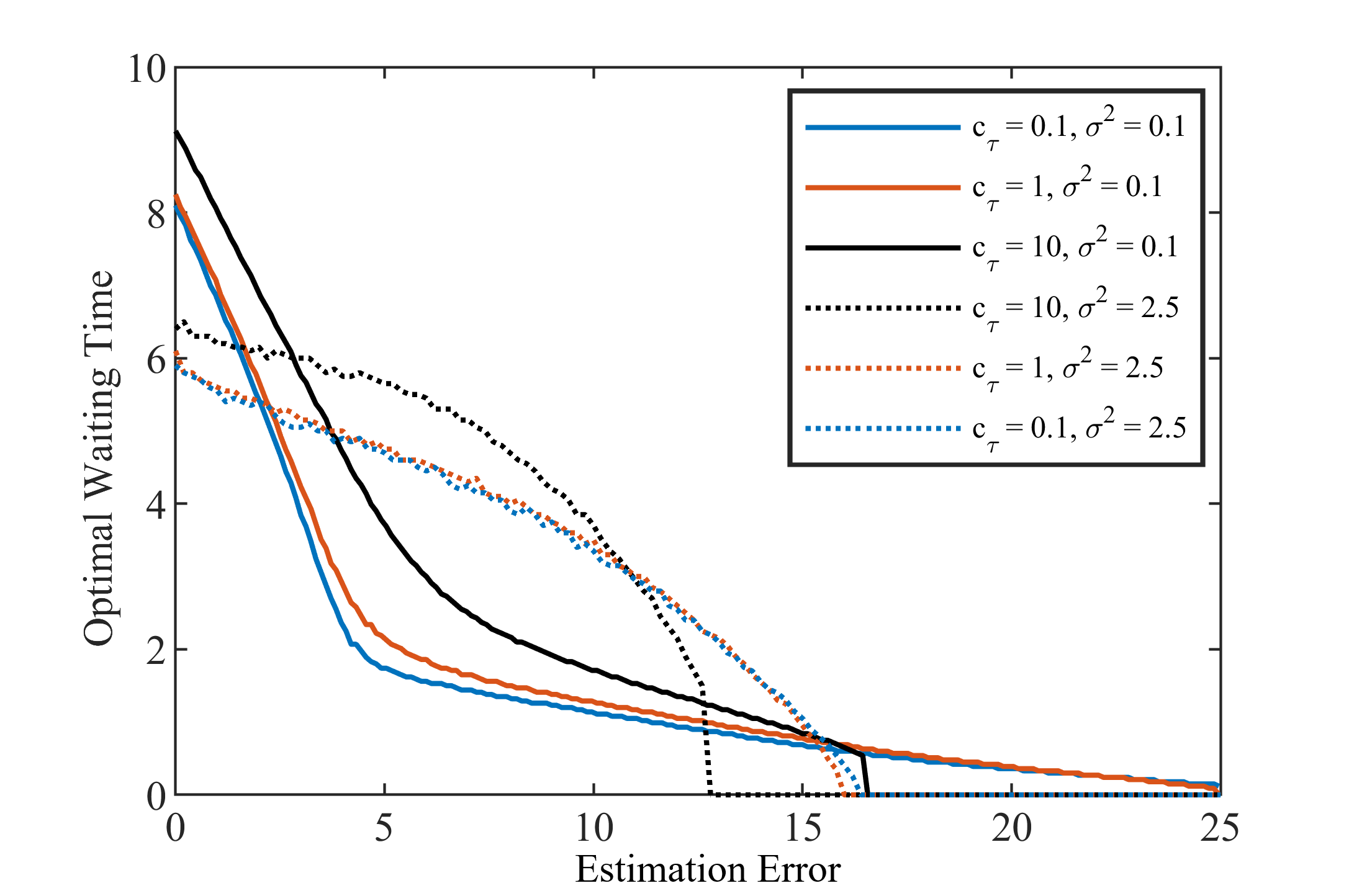}
    \caption{Optimal waiting times vs. error for $\sigma^2 = 0.1,\sigma^2 = 2.5$ and $c_s = 1$ as constant}
    \label{fig:waiting_varY_2}
\end{figure}

For low delay variance $\sigma^2 = 0.1$ the waiting time decreases with increasing error, with minor differences between $c_\tau$ values. For high delay variance $\sigma^2 = 2.5$ the waiting time decreases more sharply, with higher $c_\tau$ values resulting in significantly longer waiting times. 
When the delay variance is low, the waiting times are relatively short, and the system can afford to be more responsive to errors. The differences between the curves for different $c_\tau$ are less pronounced, indicating a smaller impact of transmission costs in low variance scenarios. On the other hand, with high delay variance, the system must wait longer to account for the increased variability in transmission times. The impact of transmission costs is more pronounced, with higher costs leading to longer optimal waiting times to balance the trade-off between error correction and cost.

For all cases, optimal waiting time decreases with increasing error. This behavior is expected as the system should react more quickly to larger errors to maintain accurate estimations. Beyond a certain error range, the optimal waiting time becomes very small or zero, indicating that for very high errors, immediate sampling is preferred, aligning with the need to correct larger deviations as promptly as possible.

\section{Conclusion}

In this paper, we addressed the problem of sampling a Wiener process under cost for remote estimation over a network connection with random IID delay. Utilizing Lagrange relaxation and iterative backward induction, we derived a policy that minimizes the time-average MSE while accounting for the costs associated with sampling and transmission. The results substantiate the efficacy of our approach in minimizing MSE under varying transmission conditions. Specifically, the robustness of the optimal policy to delay variance highlights its potential for practical implementation in communication systems where transmission delays are unpredictable. This robustness is particularly critical in applications such as remote sensing, networked control systems, and real-time monitoring, where maintaining low estimation error is paramount. Future research could explore the applicability of our findings to other types of stochastic processes and further refine the optimization techniques to accommodate more complex and dynamic network conditions.

\appendices

\section{Formulation of (\ref{eq:mseopt})} 
We derive the time-averaged cost function using renewal properties and renewal-reward theory. Let \(D_n\) denote the \(n\)-th renewal epoch, where the \(n\)-th renewal interval starts at \(D_n\) and ends at \(D_{n+1}\). The inter-renewal time is given by:
\[
X_n = D_{n+1} - D_n = Y_n + Z_{n,1} + Z_{n,2} + \cdots + Z_{n,k_n}.
\]

The reward function \(R(t)\) represents the value of cost function at time \(t\), which includes the estimation error, sampling cost, and transmission cost. We aim to find the time-average reward:
\[
\lim_{T \to \infty} \frac{1}{T} \int_{0}^{T} R(t) \, dt.
\]

The total accumulated reward up to time \(T\) can be expressed as:
\[
\int_{0}^{T} R(t) \, dt = \sum_{n=1}^{M(T)} R_n + \int_{D_{M(T)}}^{T} R(t) \, dt,
\]
where \(M(t)\) is defined in \eqref{information-set} and \(R_n\) is the accumulated reward in the \(n\)-th renewal interval :
\[
R_n = \int_{D_n}^{D_{n+1}} R(t) \, dt.
\]

Divide the total reward by \(T\) and take the limit as \(T \to \infty\):
\[
\lim_{T \to \infty} \frac{1}{T} \int_{0}^{T} R(t) \, dt = \lim_{T \to \infty} \frac{1}{T} \left( \sum_{n=1}^{M(T)} R_n + \int_{D_{M(T)}}^{T} R(t) \, dt \right).
\]

The second term, \(\frac{1}{T} \int_{D_{M(T)}}^{T} R(t) \, dt\), vanishes as \(T \to \infty\) because the duration of the partial renewal interval is bounded. Thus, the time-average reward reduces to:
\[
\lim_{T \to \infty} \frac{1}{T} \int_{0}^{T} R(t) \, dt = \lim_{T \to \infty} \frac{1}{T} \sum_{n=1}^{M(T)} R_n.
\]

Using the renewal-reward theorem, we express this as:
\[
\lim_{T \to \infty} \frac{1}{T} \int_{0}^{T} R(t) \, dt = \frac{\mathbb{E}[R_n]}{\mathbb{E}[X_n]},
\]
where:
\(\mathbb{E}[R_n]\) is the expected accumulated reward in a single renewal interval.
\(\mathbb{E}[X_n]\) is the expected inter-renewal time.

The reward in the \(n\)-th renewal interval is:
   \[
   R_n = \int_{D_n}^{D_{n+1}} (W_t - W_0)^2 \, dt + c_s k_n + c_\tau.
   \]

Using the result of renewal-reward theorem, and substituting the expressions for \(\mathbb{E}[R_n]\) and \(\mathbb{E}[X_n]\), we obtain  time average cost as:
\[ \frac
{
\sum_{n=1}^\infty \mathbb{E} \left[ \int_{D_n }^{D_{n+1}} (W_t - W_0)^2 \, dt + c_s k_n + c_\tau \right]
}
{
\sum_{n=1}^\infty \mathbb{E}[ D_{n+1}- D_n ].
}
\]

This completes the derivation of Equation \eqref{eq:mseopt} using renewal-reward theory.

\section{Proof of Lemma (\ref{lemma1})} 
Part (a) and (b) is proven in two steps:

\textbf{Step 1:} We will prove that \(\text{mse}_{\text{opt}} \leq \lambda \) if and only if \(J(\lambda) \leq 0\).

If \(\text{mse}_{\text{opt}} \leq \lambda\), then there exists a policy \(\pi = (Z_0, Z_1, \ldots) \in \Pi\) that is feasible for both (\ref{eq:mseopt}) and (\ref{eq:relaxed_problem}), which satisfies
\[
\lim_{n \to \infty} \frac{\sum_{i=0}^{n-1} \mathbb{E} \left[ \int_{D_i}^{D_{i+1}} (W_t - W_{S_i})^2 \, dt \right]}{\sum_{i=0}^{n-1} E[Y_i + T_i]} \leq \lambda. \quad 
\]
Hence,
\[
\lim_{n \to \infty} \frac{1}{n} \sum_{i=0}^{n-1} \mathbb{E} \left[ \int_{D_i}^{D_{i+1}} (W_t - W_{S_i})^2 \, dt - \lambda(Y_i + T_i) \right] \leq 0. \quad
\]
Because the inter-packet waiting times \(T_i\) are regenerative, the renewal theory \cite{ross-stochastic} tells us that the limit
\[
\lim_{n \to \infty} \frac{1}{n} \sum_{i=0}^{n-1} \mathbb{E}[Y_i + T_i]
\]
exists and is positive. By this, we get
\begin{equation}
\lim_{n \to \infty} \frac{1}{n} \sum_{i=0}^{n-1} \!\mathbb{E}\! \left[ \int_{D_i}^{D_{i+1}} \!\!\!(W_t \!-\! W_{S_i})^2 \, \!dt \!-\! \lambda(Y_i + T_i) \right] \!\! \leq \! 0.
\label{semi-pos_lagrangian}
\end{equation}

Therefore, \(J(\lambda) \leq 0\).

On the reverse direction, if \(J(\lambda) \leq 0\), then there exists a policy \(\pi = (Z_0, Z_1, \ldots) \in \Pi\) that is feasible for both (\ref{eq:mseopt}) and (\ref{eq:relaxed_problem}), which satisfies (\ref{semi-pos_lagrangian}). From there, we can derive previous steps. Hence, \(\text{mse}_{\text{opt}} \leq \lambda\). By this, we have proven that \(\text{mse}_{\text{opt}} \leq \lambda\) if and only if \(J(\lambda) \leq 0\).

\textbf{Step 2:} We need to prove that \(\text{mse}_{\text{opt}} < \lambda\) if and only if \(J(\lambda) < 0\). This statement can be proven by using the arguments in Step 1, in which “\(\leq\)” should be replaced by “\(<\)”. Finally, from the statement of Step 1, it immediately follows that \(\text{mse}_{\text{opt}} > \lambda\) if and only if \(J(\lambda) > 0\). This completes the proof of part (a) and (b).

\textbf{Part (c):} We first show that each optimal solution to (\ref{eq:mseopt}) is an optimal solution to (\ref{eq:relaxed_problem}). By the claim of the part (a) and (b), \(J(\lambda) = 0\) is equivalent to \(\text{mse}_{\text{opt}} = \lambda\). Suppose that policy \(\pi = (Z_0, Z_1, \ldots) \in \Pi\) is an optimal solution to (\ref{eq:mseopt}). Then, \(\text{mse}_{\pi} = \text{mse}_{\text{opt}} = \lambda\). Applying this in the arguments of (\ref{semi-pos_lagrangian}), we can show that policy \(\pi\) satisfies
\[
\lim_{n \to \infty} \frac{1}{n} \sum_{i=0}^{n-1} \mathbb{E} \left[ \int_{D_i}^{D_{i+1}} (W_t - W_{S_i})^2 \, dt - \lambda(Y_i + T_i) \right] = 0.
\]
This and \(J(\lambda) = 0\) imply that policy \(\pi\) is an optimal solution to (\ref{eq:relaxed_problem}).

Similarly, we can prove that each optimal solution to (\ref{eq:relaxed_problem}) is an optimal solution to (\ref{eq:mseopt}). By this, part (c) is proven.

\section{Proof of Lemma (\ref{lemma2})} 
Using the property of the Wiener process, expand \((W_t - W_0)^2\) as:
\[
(W_t - W_0)^2 = \left[(W_t - W_X) + (W_X - W_0)\right]^2.
\]
Expanding the square:
\[
\begin{split}
(W_t - W_0)^2 & = (W_t - W_X)^2 + 2(W_t - W_X)(W_X - W_0) \\
& + (W_X - W_0)^2.
\end{split}
\]

The expectation of the integral becomes:
\[
\begin{split}
\mathbb{E} \left[ \int_X^{X+Y} (W_t - W_0)^2 \, dt \right] 
&= 
\mathbb{E} \left[ \int_X^{X+Y} (W_t - W_X)^2 \, dt \right] \\
&+ 
\mathbb{E} \left[ \int_X^{X+Y} (W_X - W_0)^2 \, dt \right].
\end{split}
\]

Since \((W_t - W_X)\) is independent of \((W_X - W_0)\) and has zero mean the term:
\[
\mathbb{E} \left[ \int_X^{X+Y} 2(W_t - W_X)(W_X - W_0) \, dt \right] = 0.
\]

The term \((W_t - W_X)^2\) is the variance of the Wiener process over \([X, t]\), which simplifies to:
\[
\mathbb{E} \left[ \int_X^{X+Y} (W_t - W_X)^2 \, dt \right] = \mathbb{E} \left[ \int_X^{X+Y} (t - X) \, dt \right].
\]
Change variables \(u = t - X\) with bounds \(u \in [0, Y]\):
\[
\int_X^{X+Y} (t - X) \, dt = \int_0^Y u \, du = \frac{Y^2}{2}.
\]

The term $(W_X - W_0)^2$ is constant over \([X, X+Y]\):
\[
\mathbb{E} \left[ \int_X^{X+Y} (W_X - W_0)^2 \, dt \right] = \mathbb{E} \left[ (W_X - W_0)^2 \right] \cdot Y.
\]

As a result right hand side of the equation simplifies to:
\[
\mathbb{E} \left[ \int_X^{X+Y} (W_t - W_0)^2 \, dt \right] = 
\frac{Y^2}{2} + Y \cdot \mathbb{E}[(W_X - W_0)^2].
\]

This completes the proof.

\section{Proof of (\ref{eq:lagrangae_multiplier})} 
We start with the original problem of minimizing the long-term average mean square estimation error (MSE) subject to the costs of sampling \(c_s\) and transmission \(c_\tau\).

The original objective is to minimize the long-term average MSE:
\[
\text{mse}_{\text{opt}} \!=\! \min_{\pi \in \Pi} \lim_{N \to \infty} \frac{1}{N} \sum_{n=1}^{N} \mathbb{E} \left[ \int_{D_n}^{D_{n+1}} \!\!\!(W_t - W_0)^2 \, dt \!+\! c_s k_n \!+\! c_\tau \right]
\]

To introduce the Lagrange multiplier \(\lambda\), we add a penalty term for the transmission delays and the waiting times. This helps balance the estimation error and the associated costs. The augmented objective function becomes:
\[
J(\lambda) \!=\! \min_{\pi \in \Pi} \lim_{N \to \infty} \frac{1}{N} \sum_{n=1}^{N} C_n(\lambda).
\]
Where $C_n(\lambda)$ is the expected cost in one delivery epoch. Let \(\lambda^*\) be the value of \(\lambda\) that minimizes the augmented objective function. The goal is to find \(\lambda^*\) such that:
\[
J(\lambda^*)\!=\! \min_{\pi \in \Pi} \lim_{N \to \infty} \frac{1}{N} \sum_{n=1}^{N} C_n(\lambda^*)
\]

Since we are considering an infinite horizon, we can rewrite the objective function as an infinite sum:
\[
\sum_{n=1}^{\infty} \mathbb{E} \left[ \int_{D_n}^{D_{n+1}} (W_t - W_0)^2 \, dt \!+\! c_s k_n \!+\! c_\tau \!-\! \lambda^* (Y_n \!+\! T_n) \right] \!\geq\! 0.
\]

This ensures that the policy \(\pi\) chosen does not violate the constraint imposed by the Lagrange multiplier \(\lambda^*\). The existence of such a \(\lambda^*\) implies that there exists a balance point where the long-term average cost is minimized while ensuring that the policy adheres to the constraint. The inequality \(\sum_{n=1}^{\infty} \mathbb{E} [\cdot] \geq 0\) guarantees the feasibility of the solution.




\ifCLASSOPTIONcaptionsoff
  \newpage
\fi



%

\bibliographystyle{IEEEtran}  
\bibliography{main}   

%

\begin{IEEEbiography}[{\includegraphics[width=1in,height=1.25in,clip,keepaspectratio]{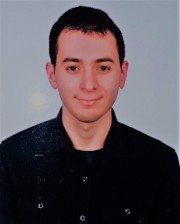}}]{Orhan T. Yavascan}
received the B.S. degree in electrical and electronics engineering from Bilkent University in 2019. He is currently pursuing the M.S. degree with Middle East Technical University (METU). His current research interests include communication networks and information freshness.
\end{IEEEbiography}

\begin{IEEEbiography}
[{\includegraphics[width=1in,height=1.25in,clip,keepaspectratio]{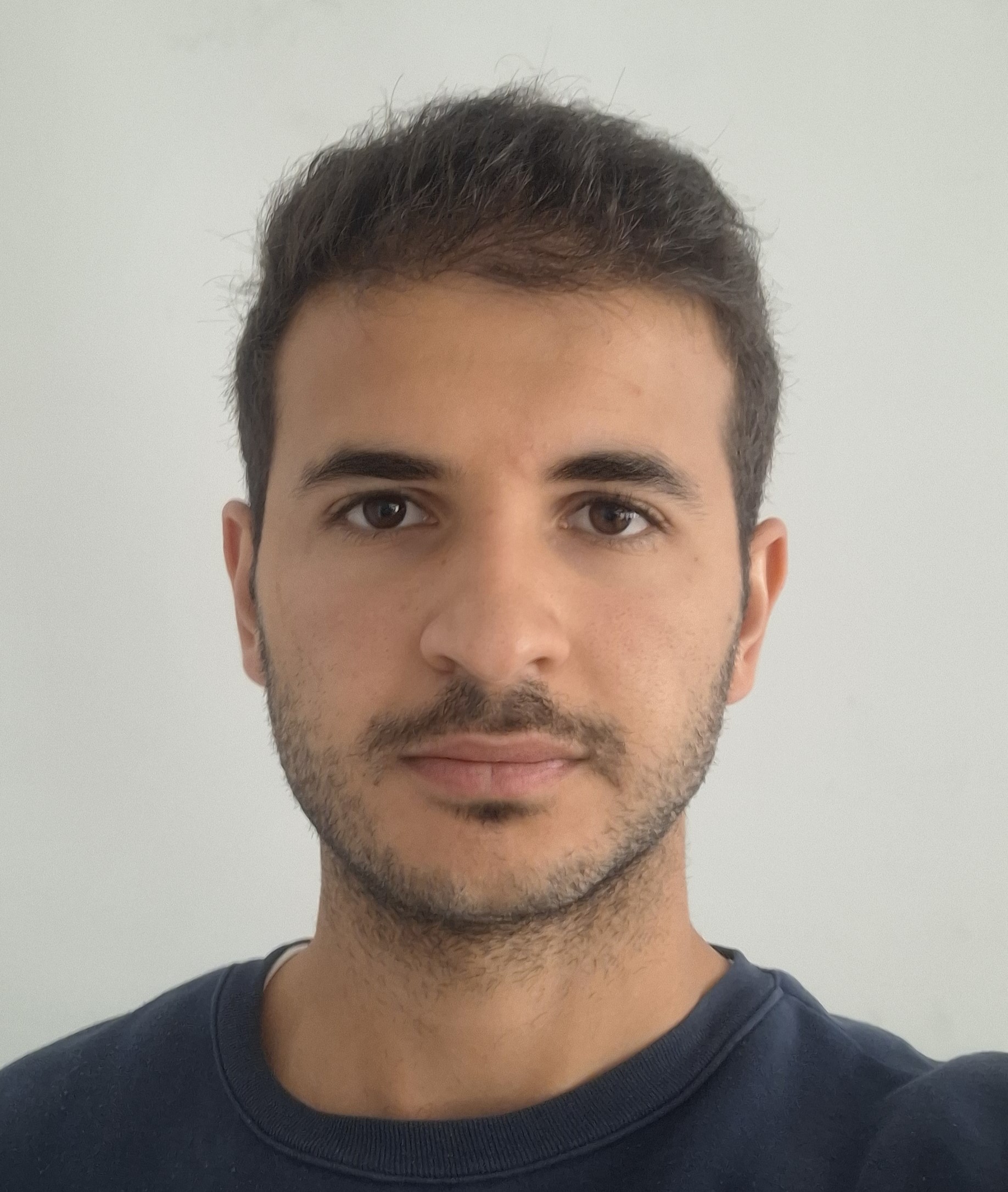}}]
{Süleyman Çıtır}
received the B.S. degree in Electrical and Electronics Engineering from Middle East Techincal University (METU) in 2023. He is currently pursuing the M.S. degree with METU. His current research interests include non-terrestrial networks.
\end{IEEEbiography}

\vspace{0em}

\begin{IEEEbiography}
[{\includegraphics[width=1in,height=1.25in,clip,keepaspectratio]{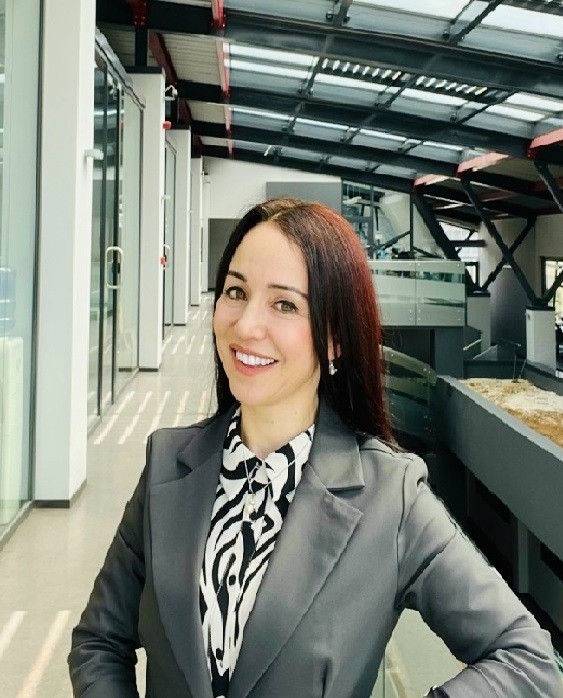}}]
{Elif Uysal}
(S’95-M’03-SM’13-F'24) is a Professor in the Department of Electrical
and Electronics Engineering at the Middle East Technical University (METU),
in Ankara, Turkey. She received the Ph.D. degree in EE from Stanford
University in 2003, the S.M. degree in EECS from the Massachusetts Institute
of Technology (MIT) in 1999 and the B.S. degree from METU in 1997.
From 2003-05 she was a lecturer at MIT, and from 2005-06 she was an
Assistant Professor at the Ohio State University (OSU). Since 2006, she has been with METU, and held visiting positions at OSU and MIT during 2014-2016. Her research interests are at the junction of communication and networking theories, with recent application to Space and Interplanetary
Networking. Dr. Uysal is a recipient of a 2024 ERC Advanced Grant, 2014 Young Scientist Award from the Science Academy of Turkey, an IBM Faculty Award (2010), the Turkish National Science Foundation Career Award (2006), an NSF Foundations on Communication research grant (2006-2010), the MIT Vinton Hayes Fellowship, and the Stanford Graduate Fellowship. She has served as associate editor for 
the IEEE/ACM Transactions on Networking, and the IEEE Transactions on Wireless Communication. 
\end{IEEEbiography}




\end{document}